\documentclass[journal]{IEEEtran}

\usepackage{setspace}
\usepackage{lettrine}

\usepackage{graphics,
           psfrag,
           epsfig,
           amsthm,
           cite,
           amssymb,
           url,
           dsfont,
           subfigure,
           algorithm,
           algorithmic,
           balance,
           enumerate,
           color,
           setspace
}
\usepackage{amsmath}

\usepackage{epstopdf}

\newtheorem{definition}{Definition}

\newtheorem{lemma}{Lemma}

\newtheorem{theorem}{Theorem}

\newtheorem{preposition}{Preposition}

\newcommand{\eref}[1]{(\ref{#1})}
\newcommand{\sref}[1]{Section~\ref{#1}}

\newcommand{\appref}[1]{Appendix~\ref{#1}}
\newcommand{\fref}[1]{Figure~\ref{#1}}

\newcommand{\cref}[1]{Constraint~\ref{#1}}
\newcommand{\thref}[1]{Theorem~\ref{#1}}

\newcommand{\lref}[1]{Lemma~\ref{#1}}

\newcommand{\ignore}[1]{}

\IEEEoverridecommandlockouts
\begin{document}

\title{\vspace{-0.5cm}Data Dissemination using Instantly Decodable Binary Codes in Fog-Radio Access Networks}

\author{
   \authorblockN{Ahmed Douik, \textit{Student Member, IEEE} and Sameh Sorour, \textit{Senior Member, IEEE}}
    
    {\thanks {    
Ahmed Douik is with the Department of Electrical Engineering, California Institute of Technology, Pasadena, CA 91125 USA (e-mail: ahmed.douik@caltech.edu).

Sameh Sorour is with the Department of Electrical and Computer Engineering, University of Idaho, Moscow, ID 83844, USA (e-mail: samehsorour@uidaho.edu).
}}
\vspace{-0.8cm}    }

\maketitle

\IEEEoverridecommandlockouts

\begin{abstract}
Consider a device-to-device (D2D) fog-radio access network wherein a set of devices are required to store a set of files. Each device is connected to a subset of the cloud data centers and thus possesses a subset of the data. This paper investigates the problem of disseminating all files among the devices while reducing the total time of communication, i.e., the completion time, using instantly decodable network coding (IDNC). While previous studies on the use of IDNC in D2D systems assume a fully connected communication network, this paper tackles the more realistic scenario of a partially connected network in which devices can only target devices in their transmission range. The paper first formulates the optimal joint optimization of selecting the transmitting device(s) and the file combination(s) and exhibits its intractability. The completion time is approximated using the celebrated decoding delay approach by deriving the relationship between the quantities in a partially connected network. The paper introduces the cooperation graph and demonstrates that the relaxed problem is equivalent to a maximum weight clique problem over the newly designed graph wherein the weights are obtained by solving a similar problem on the local IDNC graphs. Extensive simulations reveal that the proposed solution provides noticeable performance enhancement and outperforms previously proposed IDNC-based schemes.
\end{abstract}


\section{Introduction} 

\lettrine[lines=2]{W}{ith} their increased popularity and their abundance in radio access networks, smartphones are becoming active parts of the system. While traditional cell phones were considered as mere terminals by the service providers, the increased computation power and storage capacity of smartphones are turning them into active components of the network. For example, to support a massive number of devices and further reduce latency, the notion of device-to-device (D2D) \cite{6805125} communication has been proposed as a potential candidate for the next generation mobile radio system ($5$G) \cite{6736746}. Similarly, multiple works investigate the use of smartphones are potential relays in the network. 

Lately, the notion of fog computing \cite{7513863,7537176} emerged as a new paradigm in radio access network in which not only the communication and computing resources of the mobile devices are exploited but also their storage capacity. Such paradigm shift allows not only to save the data center resources but also to have fast access to the files and thus to meet the ever increasing data rates and the Quality of Service (QoS) requirements \cite{6824752}. As for all new notions, the definition of fog networking and computing is still ambiguous and does not make consensus in the literature, e.g., \cite{7727971,7776570,7558153}. This paper considers fog computing from the storage perspective in which the cloud data centers disseminate files in the network for faster access.

By exploiting the computation abilities of the intermittent nodes in the network, Network Coding (NC) has shown remarkable abilities in significantly improving the network capacity and reducing the delay of wireless broadcast configurations \cite{19518321}. For D2D systems in which devices exchange packets over a short range and possibly more reliable channels, NC is a suitable complementary solution \cite{6404743} to provide reliable and secure data communications over ad-hoc networks such that Internet of Things (IoT) and wireless sensor networks.

While random NC schemes require computationally expensive matrix inversion, Instantly Decodable Network Coding (IDNC) \cite{5425315} is an important subclass of NC that is suitable for battery-powered D2D communications. IDNC provides an incredibly fast, or as it name indicates instantly, encoding and decoding through simple binary XOR operations which are particularly well adapted for the network of interest in this paper wherein devices are highly limited in terms of computation complexity. Besides, IDNC provides progressive decoding which is a fundamental feature that makes files ready-to-use from their reception instant. For its aforementioned benefits, IDNC is employed in various settings \cite{5425315,4313060,20112430,6570827,6620795,6503457,6970860}.

Consider a D2D fog-radio access network (F-RAN) wherein a set of devices are required to store a set of files. Each device is connected to a subset of the cloud data centers and thus possesses a subset of the data. This paper investigates the problem of disseminating all files among the devices while reducing the total time of communication, i.e., the completion time using IDNC. While previous studies on the use of IDNC in D2D systems assume a fully connected communication network, this paper tackles the more realistic scenario of a partially connected network in which devices can only target devices in their transmission range. However, the assumption of a global coordinator in the network is preserved and can be alleviated in future work using a game theoretical approach similar to the one proposed in \cite{14122014}.

Reducing the number of transmissions is intractable due to the dynamic nature of the channels. Various approximations of the completion time have been suggested in the literature among which the decoding delay approach in \cite{13051412} that allowed reducing the completion time below its best-known bound. Similar to the completion time, finding the optimal schedule for decoding delay minimization is intractable \cite{6503457,5205612}. However, the authors in \cite{6620795} propose an efficient on-line decoding delay minimization scheme. This paper suggests using a similar approach by deriving the relationship between the completion time and the decoding delay in the partially connected network of interest and using the decoding delay expressions provided in \cite{5456269} to obtain an online completion time minimization scheme.

The paper's main contribution is to propose an efficient method for disseminating the files among the devices on a partially connected D2D F-RAN. The joint optimization over the set of transmitting devices and data combinations so as to reduce the number of transmissions is first formulated and shown to be intractable. Due to the intractability of the completion time, the paper proposes approximating the metric by deriving its relationship to the decoding delay. Finally, using the expressions of the decoding delay in the literature, the paper designs the cooperation graph and shows that the relaxed problem is equivalent to a multi-layer maximum weight clique problem. Simulation results reveal that the proposed solution outperforms previously IDNC-based schemes in partially connected communication systems. Due to space limitation, proofs, and additional simulations can be found in the online technical report \cite{5684518}.

The rest of this paper is organized as follows: \sref{sec:sys} presents the system model and relevant definitions. The completion time is expressed, and the problem is formulated in \sref{sec:com}. In \sref{sec:int}, the cooperation graph is constructed, and the solution to the collision-free scenario is suggested. The solution is extended to the general cooperation in \sref{sec:gen} Finally, before concluding in \sref{sec:con}, simulation results are provided in \sref{sec:sim}.

\section{System Model and Definitions} \label{sec:sys}

\subsection{System Model and Parameters}

Consider a D2D F-RAN consisting of a set $\mathcal{U}$ of $U$ devices. Initially available at the data centers, the central controller aims to store $F$ files (denoted by the set $\mathcal{F}$) at all devices. Each device $u$ is connected to some data centers from which it obtained a subset of the files $\mathcal{H}_u \subseteq \mathcal{F}$. Call the missing files at the $u$-th device its Wants set and denote it by $\mathcal{W}_u$. The central controller aims to design a transmission protocol so that each device receives a copy of all its missing files. For devices to be able to receive all data, each file is assumed to be possessed by at least one device, i.e., $f \in \bigcup_{u \in \mathcal{U}} \mathcal{H}_u, \ \forall \ f \in \mathcal{F}$.

The network topology is captured by a symmetric, unit diagonal, $U \times U$ connectivity matrix $\mathbf{C}=[c_{uu^\prime}]$ wherein the entry $c_{uu^\prime}$ is equal to $1$ if devices $u$ and $u^\prime$ are in the transmission range of each others. Furthermore, the paper assumes that the network (and thus the matrix $\mathbf{C}$) is connected, i.e., each device can target any other device through a single or a multi-hop transmission. If some part of the network is disconnected, it can be considered as an independent network and optimized separately. The coverage zone $\mathcal{C}_u$ of the $u$-th device is defined as the set of devices in the transmission range of the $u$-th device. In other words, $\mathcal{C}_u$ is defined by:
\begin{align}
\mathcal{C}_u = \{ u^\prime \in \mathcal{U} \ | \ c_{uu^\prime}=1\}.
\end{align}

The paper considers that the D2D transmissions are subject to independent but not necessarily identical erasures. The erasure probabilities are represented by the zeros diagonal $U \times U$ matrix $\mathbf{E}=[\epsilon_{uu^\prime}]$ wherein $\epsilon_{uu^\prime}$ represents the probability that the transmission from the $u$-th device is erased at the $u^\prime$-th device. The erasures probabilities are assumed to be known by the central unit and to remain constant during the transmission of a single file combination. Due to the asymmetry of the channels and to a potential difference in the transmit power, the erasures $\epsilon_{uu^\prime}$ and $\epsilon_{u^\prime u}$ are not necessarily equal.

Devices cooperate to complete the reception of all files by exchanging XOR-encoded files to devices in their transmission range.  This paper assumes the central controller has full knowledge of the distribution of lost and received files at each device which can be accomplished by the exchange of positive and negative acknowledgments (ACKs and NACKs) through a dedicated feedback channel. 

\subsection{Definitions and Notation}

This subsection gathers the relevant definitions and notation used throughout the paper. Let $\mathcal{S}$ denote a schedule formed from the transmitting devices and the file combination for each time slot. The paper aims to find the schedule $\mathcal{S}$ that minimizes the total number of transmissions, known as the completion time $\mathcal{T}$ and defined as follows:
\begin{definition}[Individual Completion Time]
The individual completion time $\mathcal{T}_u(\mathcal{S})$ of the $u$-th device experienced by following the schedule $\mathcal{S}$ is the number of transmissions required until the device obtains all its missing files.
\end{definition}
\begin{definition}[Completion Time]
The overall completion time $\mathcal{T}(\mathcal{S})$ experienced by following the schedule $\mathcal{S}$ is the number of transmissions required until all devices obtains their missing files. In other words, $\mathcal{T}(\mathcal{S}) = \max_{u \in \mathcal{U}} \mathcal{T}_u(\mathcal{S})$
\end{definition}

\begin{figure}[t]
\centering
\includegraphics[width=0.4\linewidth]{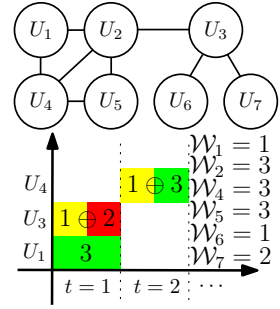}\\
\caption{An example of a schedule in a partially connected D2D-enable network composed of $7$ devices and $3$ files. In the first time slot both devices $U_1$ and $U_3$ transmit. In the second time slot device $U_4$ transmits.}\label{fig:network}
\end{figure}

\fref{fig:network} represents an example of a schedule in a partially connected D2D F-RAN. Unlike fully connected D2D systems, multiple devices are able to transmit simultaneously. The individual completion time of devices $U_6$ and $U_7$ is one unit. However, the overall completion time is $2$ units so as to satisfy all devices.

Inspired by the work in \cite{13051412}, this paper employs a decoding delay approach to efficiently reduce the completion time. To define the decoding delay, first introduce the different reception options for the $u$-th device as follows:
\begin{itemize}
\item \emph{Instantly Decodable:} A file is instantly decodable if it allows the device to recover one of its missing files. Given that encoding is based solely on XOR operations, a file combination is instantly decodable if it contains exactly one file from $\mathcal{W}_u$.
\item \emph{Not Instantly Decodable:} A file is not instantly decodable if it is either non-decodable or previously received. Hence, a file is not instantly decodable if it does not contain exactly a single file from $\mathcal{W}_u$.
\end{itemize}

The decoding delay \cite{5456269} is defined as follows:
\begin{definition}[Decoding Delay]
The decoding delay $\mathcal{D}_u(\mathcal{S})$ of the $u$-th device, with non-empty Wants sets, increases by one unit if the device is not able to hear exactly a single transmission or if it hears a non-instantly decodable file combination in the schedule $\mathcal{S}$.
\end{definition}

The rationale behind the definition of the decoding delay is that it only accounts for delays caused by the transmitting devices and the coding decisions. Hence, the channel erasures are not considered in the definition. In \fref{fig:network}, the decoding delay is computed as follows:
\begin{itemize}
\item Device $U_1$ experience one unit of delay as it transmits in the first time slot. In the second time slot, it receives an instantly decodable file and hence no additional delay. 
\item Device $U_2$ and $U_5$ experience one unit of delay as $U_2$ is in interference and $U_5$ is out of the transmission range during the first time slot. Hence, they cannot hear exactly one transmission.
\item Device $U_3$ and $U_4$ do not experience any delay as $U_3$ has an empty Wants set and $U_4$ also has an empty Wants set when it transmits.
\end{itemize}

The notation used in the paper are the following. Matrices are represented by bold upper case characters, e.g., $\mathbf{X}$. The entry at the $i$-th row and $j$-th column of $\mathbf{X}$ is denoted by $x_{ij}$. Sets are indicated by calligraphic letters, e.g., $\mathcal{X}$. The notation $\overline{\mathcal{X}}$ and $|\mathcal{X}|$ represent the complement and the cardinal of the set $\mathcal{X}$. The power set of $\mathcal{X}$ is represented by $\mathcal{P}(\mathcal{X})$. 

\section{Completion Time Expression} \label{sec:com}

This section first formulates the completion time minimization problem in partially connected D2D F-RAN. Due to the intractability of minimizing the completion time, the section proposes approximating the completion time by a more tractable metric known as the anticipated completion time that matches the genuine completion time for a large number of decoding delay-free transmissions. Using the anticipated completion time, the problem is reformulated as a joint online optimization over the set of transmitting devices and file combinations.

\subsection{Problem Formulation and Completion Time Expression}

The completion time problem is the one of finding the set of transmitting devices and file combinations for each time slot so as to minimize the number of transmissions. Formally, the problem is expressed as follows:
\begin{align}
\mathcal{S}^* &= \arg\min_{\mathcal{S} \in \mathbf{S}} \mathcal{T}(\mathcal{S}) = \arg\min_{\mathcal{S} \in \mathbf{S}} \left\{\underset{u \in \mathcal{U}}{\text{max }} \mathcal{T}_u(\mathcal{S})\right\},
\label{eq:1}
\end{align}
where $\mathbf{S}$ is the set of all feasible schedules. The optimization problem \eref{eq:1} is intractable as it depends on future channel realizations. Furthermore, even for erasure-free transmission, the search space $\mathbf{S}$ is prohibitively huge for any moderate sized network and number of files \cite{6775017}. 

In order to efficiently reduce the completion time with a reasonable complexity, the following lemma suggests re-expressing it using an expression involving the decoding delay:
\begin{lemma}
The individual completion time $\mathcal{T}_u(\mathcal{S}$ of the $u$-th can be expressed as follows:
\begin{align}
\mathcal{T}_u(\mathcal{S}) = |\mathcal{W}_u| + \mathcal{D}_u(\mathcal{S}) + \mathcal{E}_u(\mathcal{S}),
\end{align}
where $\mathcal{E}_u(\mathcal{S})$ is the number of erased files during the transmission of the schedule $\mathcal{S}$.
\label{l1}
\end{lemma} 

\begin{proof}
The lemma is demonstrated by identifying all possible transmissions at the $u$-th device and translating their effect on the completion time. The complete proof can be found in \appref{app1}.
\end{proof}

The following theorem exploits the expression in \lref{l1} to approximate the completion time.
\begin{theorem}
The individual completion time experienced by the $u$-th device after the transmission of the schedule $\mathcal{S}$ can be approximated as follows:
\begin{align}
\mathcal{T}_u(\mathcal{S}) = \cfrac{|\mathcal{W}_u| + \mathcal{D}_u(\mathcal{S}) - \mathds{E}[\epsilon_u]} {1-\mathds{E}[\epsilon_u]},
\label{eq:3}
\end{align}
where $\mathds{E}[\epsilon_u]$ is the expected erasure experienced by the $u$-th device from all transmitting ones. Although one can update online the expected erasure after each transmission, simulation results reveal that an efficient approximation of the expected erasure is given by $\overline{\epsilon}_u = \cfrac{ \sum\limits_{u^\prime \in \mathcal{C}_u}\epsilon_{u^\prime u }}{|\mathcal{C}_u|}$ as it represents a scenario in which all devices are equality likely to transmit.
\label{th1}
\end{theorem}

\begin{proof}
To show the theorem, the mean expression of the completion time is expended using the expected erasure probability. With such expression, the probability distribution of the erased files is derived, and their sum approximated using the law of large numbers. Finally, the expected erasure probability is shown to coincide with the $\overline{\epsilon}_u$ for collision-free transmission and uniform distribution of transmitting devices. The complete proof can be found in \appref{app2}.
\end{proof}

The rest of the paper uses the approximation in \thref{th1} as an equality as it holds for a large number of files, devices, and the collision-free scenario under investigation in this paper. It also neglects the set of devices out of the transmission range as the assumption holds for moderately connected networks.

\subsection{Online Completion Time Reduction}

The completion time reduction problem can be approximated using the expression provided in \thref{th1} as follows:
\begin{align}
\mathcal{S}^* &= \arg\min_{\mathcal{S} \in \mathbf{S}} \left\{\underset{u \in \mathcal{U}}{\text{max }} \cfrac{|\mathcal{W}_u| + \mathcal{D}_u(\mathcal{S}) - \mathds{E}[\epsilon_u]}{1-\mathds{E}[\epsilon_u]}\right\}
\end{align}

Even though the above expression is challenging to optimize, it allows to conclude that the term that affects the most the completion time is the maximum decoding delay. Hence, the philosophy of the proposed online solution is to reduce the probability of increase in the anticipated completion time defined at the $t$-th transmission as follows:
\begin{align}
\mathcal{T}_u(t) = \cfrac{|\mathcal{W}_u| + \mathcal{D}_u(t) - \mathds{E}[\epsilon_u]}{1-\mathds{E}[\epsilon_u]},
\label{eq:2}
\end{align}
where $\mathcal{D}_u(t)$ is the decoding delay experienced until the $t$-transmission. Note that the anticipated completion time \eref{eq:2} matches the genuine completion time \eref{eq:3} if the device does not experience any additional decoding delay for the remaining transmissions.

To reduce the probability of increase of the anticipated completion time, define $\mathcal{L}(t)$ as the critical devices \footnote{The set $\mathcal{L}(t)$ is denoted by $\mathcal{L}$ in the rest of the paper with the convention that the missing index represents the current transmission.} as those that can potentially increase the maximum anticipated completion time at the $t$-th transmission as: 
\begin{align}
&\mathcal{L}= \left\{u \in \tilde{\mathcal{U}} \bigg| \mathcal{T}_u(t-1) + \frac{1}{1-\mathds{E}[\epsilon_u]} \geq  \max_{u^\prime \in \tilde{\mathcal{U}}}\left(\mathcal{T}_{u^\prime}(t)\right)\right\}, \nonumber
\end{align}
where $\tilde{\mathcal{U}}$ is the set of devices with non-empty Wants set.

Let $\mathcal{A} \in \mathcal{P}(\mathcal{U})$ denote the set of transmitting devices, $\mathcal{I} \subset \tilde{\mathcal{U}}$ the set of devices in collision, i.e., can hear multiple transmissions, and  $\mathcal{J} \subset \tilde{\mathcal{U}}$ the set of devices out of the transmission range of the transmitting devices\footnote{Variables $\mathcal{I}$ and $\mathcal{J}$ are function of the set of transmitting devices. However, for ease of notation, the set $\mathcal{A}$ is dropped.}. The joint optimization over the transmitting devices $a \in \mathcal{A}$ and their file combinations $\kappa_a(\mathcal{A})$ is given in the following paper's main theorem:
\begin{theorem}
The set of transmitting devices and their file combinations that minimizes the probability of increase in the anticipated completion time is the solution to the following joint optimization problem:
\begin{subequations}
\label{eq:4}
\begin{align}
\mathcal{A}^* &= \arg \max_{ \mathcal{A} \in \mathcal{R}(\mathcal{L}) } \sum\limits_{a \in \mathcal{A}} y(\kappa_{a}^*(\mathcal{A})) \\
\text{subject to } \kappa_{a}^*(\mathcal{A}) &= \arg \max_{\kappa_{a}(\mathcal{A}) \in \mathcal{P}(\mathcal{H}_{a})} y(\kappa_{a}(\mathcal{A}))  \\
y(\kappa_{a}(\mathcal{A})) &= \sum\limits_{u \in \mathcal{L} \cap \tau_{a}(\kappa_{a}(\mathcal{A}))} \log \cfrac{1}{\epsilon_{au}}  
\end{align} 
\end{subequations}
where $\tau_{a}(\kappa_{a}^*(\mathcal{A}))$ is the set of targeted devices when device $a$ transmits the file combination $\kappa_{a}$, and $\mathcal{R}(\mathcal{L})$ is the set of feasible cooperation defined as follows:
\begin{align}
&\mathcal{R}(\mathcal{L}) = \{ \mathcal{A} \in \mathcal{P}(\mathcal{U}) | \mathcal{L} \cap (\mathcal{A} \cup \mathcal{I} \cup \mathcal{J} ) = \varnothing \}
\label{eq:8}
\end{align}
\label{th2}
\end{theorem}

\begin{proof}
The theorem is shown by expressing the probability of an increase in the anticipated completion time. The joint optimization over the set of transmitting devices and the file combinations is formulated. Using the definition of the critical set and the network topology, the problem can be reformulated as a constrained optimization wherein the objective function represents the set of transmitting devices and the constraint the file combinations. Finally, using the expression of the decoding delay provided in \cite{5456269}, the optimization problem is explicitly formulated. The complete proof can be found in \appref{app3}.
\end{proof}

\section{Collision-free Solution} \label{sec:int}

This section proposes solving the optimization problem in \thref{th2} in the particular scenario of cooperation without collision. It suggests choosing the limited set of transmitting devices and file combination in such a way that minimize the likelihood of an increase of the anticipated completion time. In particular, the section shows that the global solution of \eref{eq:4} can be efficiently reached when imposing restrictions on the set of transmitting devices through a graphical formulation. The relaxed completion time problem is shown to be equivalent to a maximum weight clique problem in the cooperative graph wherein the weight of each vertex is obtained by solving a multi-layer maximum weight clique in the local IDNC graph. 

\subsection{Problem Relaxation}

Due to the high interdependence between the optimization variables in \eref{eq:4}, both problems cannot be solved separately. This is mainly due to collision at certain devices upon which depends the optimal file combination. Collision occurs in a scenario wherein transmitting devices has a non-empty intersection of coverage zone. Hence, this section focuses on cooperation without collision, i.e., $\mathcal{I} = \varnothing$. The set $\mathcal{N}$ of such cooperation can be expressed as:
\begin{align}
\mathcal{N} = \{ \mathcal{A} \in \mathcal{P}(\mathcal{U}) \ | \ \mathcal{C}_u \cap \mathcal{C}_{u^\prime} \cap \tilde{\mathcal{U}} = \varnothing, \forall \ (u,u^\prime) \in  \mathcal{A}\}.
\label{eq:6}
\end{align}

Even tough the problem is challenging to solve, by imposing extra limitation on the possible cooperation, the optimization problem becomes more tractable. Indeed, under the collision-free cooperation constraint \eref{eq:6}, the optimization variables can be decoupled as shown in the following preposition.
\begin{preposition}
Under the cooperation limitation \eref{eq:6}, the optimal set of transmitting devices and their optimal file combination can be expressed as follows:
\begin{subequations}
\label{eq:7}
\begin{align}
\mathcal{A}^* &= \arg \max_{ \mathcal{A} \in \mathcal{N} \cap \mathcal{R}(\mathcal{L}) } \sum\limits_{a \in \mathcal{A}} y(\kappa_{a}^*) \label{eq:7a} \\
\text{with } y(\kappa_{a}^*) &=  \max_{\kappa_{a} \in \mathcal{P}(\mathcal{H}_{a})} \sum\limits_{u \in \mathcal{L} \cap \tau_{a}(\kappa_{a})} \log \cfrac{1}{\epsilon_{au}} \label{eq:7b}
\end{align} 
\end{subequations}
\end{preposition}

As shown in \eref{eq:7a} and \eref{eq:7b}, the optimization problem are decoupled which allows solving efficiently each separately. Equation \eref{eq:7b} translates the contribution of device $a$ to the network and equation \eref{eq:7a} optimizes the sum of the contribution under cooperation restrictions.

\subsection{Proposed Solution}

As shown in the previous subsection, the optimization problem \eref{eq:7b} reflects the contribution of the transmitting $u$-th device to the network. The problem can be efficiently solved using the multi-layer local IDNC graph formulation suggested in \cite{13051412}\footnote{While the connectivity conditions of the local graph formulation are identical to the one proposed in \cite{13051412}, the number of vertices, their weight, and the layer separation are different as shown in \appref{app4}.}. The multi-layer graph is generated by associating each device $u^\prime$ in the transmission range of the $u$-th device and a missing file $f$ to a vertex $v_{u^\prime f}$. Vertices are connected by an edge if the resulting file combination is instantly decodable for both devices represented by the vertices.

The set of transmitting devices, i.e., problem \eref{eq:7a}, is chosen by using a modified version of the cooperation graph introduced in \cite{5456269}. The cooperation graph $\mathcal{G}(\mathcal{V},\mathcal{E})$ is build by generating a vertex of each non-critical device in the network, i.e., $\mathcal{V} = \mathcal{N} \cap \mathcal{R}(\mathcal{L})$. Two device are connected if their satisfy the cooperation restriction in \eref{eq:6} which implies the condition \eref{eq:8}. In other words, vertices $v_u$ and $v_{u^\prime}$ are connected by an edge in $\mathcal{E}$ if the following condition holds:
\begin{align}
\mathcal{C}_{u} \cap \mathcal{C}_{u^\prime} \cap \tilde{\mathcal{U}} = \varnothing.
\end{align}

The following theorem reformulates the completion time reduction problem in collision-free scenarios as a graph theory problem over the cooperation graph.
\begin{theorem}
The optimal solution to the joint completion time optimization problem \eref{eq:7} under the collision-free restriction is equivalent to the maximum weight clique problem in the cooperation graph wherein the weight of each vertex $v_u$ is given by:
\begin{align}
w(v_u) = \sum\limits_{u^\prime \in \mathcal{L} \cap \tau_{u}(\kappa_{u})} \log \cfrac{1}{\epsilon_{uu^\prime}},
\end{align}
and $\kappa_{u}$ is obtained by solving the maximum weight clique problem in the multi-layer local graph of the $u$-device wherein the weight of each vertex $v_{u^\prime f}$ is:
\begin{align}
w(v_{u^\prime f}) = - \log (\epsilon_{u u^\prime}).
\end{align}
\label{th3}
\end{theorem}

\begin{proof}
The theorem is shown reformulating both problems \eref{eq:7a} and \eref{eq:7b} as graph theory problems in the cooperation and local IDNC graph, respectively. The steps in showing that \eref{eq:7b} correspond to the maximum weight clique in the local IDNC graph are similar to the ones used in \cite{13051412}. Afterward, the proof establishes a one-to-one mapping between the set of possible cooperation and the set of cliques in the cooperation graph. Finally, given that the weight of each clique corresponds to the objective function in \eref{eq:7a}, the paper concludes that the optimal solution is the maximum weight clique. A complete proof can be found in \appref{app4}.
\end{proof}

\section{General Solution} \label{sec:gen}

This section proposes extending the completion time reduction solution by relaxing the interference-free constraint of \sref{sec:int}. The fundamental concept in finding the optimal solution to the joint optimization problem proposed in \thref{th2} is to extend the cooperation graph with clusters of devices such that the collaboration between these ``virtual devices" is interference-free. Afterward, using the proposed interference-free solution in the extended graph generates the optimal solution to the joint optimization problem. Finally, due to the high complexity of finding the optimal solution which results from the huge number of potential clusters, the section suggests a lower complexity algorithm that generates only a subset of the virtual devices.

\subsection{Extended Cooperation Graph}

The extended cooperation graph is introduced in \cite{9745154} to discover the optimal solution to the decoding delay reduction problem in partially connected D2D-enabled networks. The graph, while it allows representing all cooperation between devices, ensures that the optimal file combination each device (genuine or virtual) can make only depends on that device. Hence, the formulation allows the separation of the set of transmitting devices and the file combination as the interference-free constraint in \sref{sec:int}. The first part of this subsection describes the set of feasible clusters that satisfy the constraints. The second part construct the extended cooperation graph.

Let $\mathbf{Z}$ be the set of feasible clusters. In order to have a compact and feasible representation, elements in $\mathcal{Z} \in \mathbf{Z}$ need to verify the following constraints:
\begin{itemize}
\item The feasibility of the set of transmitting device is given by $\mathcal{Z}$ is included $\mathcal{N} \cap \mathcal{R}(\mathcal{L})$.
\item The minimum representation of clusters such that the optimal file combination of each cluster only depends on the clusters is given by $\mathcal{C}^T(Z) \cap \mathcal{C}^T(\mathcal{Z} \setminus Z) \neq \varnothing, \forall \ Z \subset \mathcal{Z}$
\end{itemize}
Therefore, the set $\mathbf{Z}$ is constructed as follows:
\begin{align}
\mathbf{Z} =& \{ \mathcal{Z} \in \mathcal{P}(\mathcal{N} \cap \mathcal{R}(\mathcal{L})) \ |  \nonumber \\
& \qquad \ \mathcal{C}^T(Z) \cap \mathcal{C}^T (\mathcal{Z} \setminus Z) \neq \varnothing,\ \forall \ Z \subset \mathcal{Z}\}.
\end{align}

Given the set of feasible clusters $\mathbf{Z}$, the construction of the extended cooperation graph follows similar steps than the construction of the cooperation graph. A vertex $v$ is generated for each cluster $ \mathcal{Z} \in \mathbf{Z}$. Vertices $v$ and $v^\prime$ are connected if their coverage zone are disjoint wherein the coverage zone of a cluster is defined as the union of the coverage zones of its devices.

\subsection{Completion Time Reduction}

Let the extended multi-layer $\mathcal{G}_u(\mathcal{Z})$ be generated as the multi-layer IDNC graph at the exception that the transmitting devices and the one in interference in $\mathcal{Z}$ are omitted during the vertices generation phase. The connectivity conditions are the same as the local IDNC graph. In other words, the extended graph represents the possible file combinations and targeted devices for a set of interfering transmitting devices.

The following theorem reformulates the joint optimization problem in \thref{th2} as a maximum weight clique search over the extended cooperation graph.
\begin{theorem}
The optimal solution to the joint optimization over the set of transmitting devices and the file combination \eref{eq:4} is equivalent to a maximum weight clique in the extended cooperative graph wherein the weight vertex $v$ corresponding to cluster $\mathcal{Z}$ is given by:
\begin{align}
w(v) = \sum_{u \in \mathcal{Z}} \sum_{u^\prime \in \tau_u(\kappa_u(\mathcal{Z}))} \cfrac{1}{\log(\epsilon_{uu^\prime})}
\end{align}
and $\kappa_u(\mathcal{Z})$ is obtained by solving the maximum weight clique problem in the extended multi-layer IDNC graph $\mathcal{G}_u(\mathcal{Z})$ wherein the weight of vertex $v_{u^\prime f}$ is:
\begin{align}
w(v_{u^\prime f}) = - \log (\epsilon_{u u^\prime}).
\end{align}
\label{th4}
\end{theorem}

\begin{proof}
The theorem is established by showing a one-to-one mapping between all feasible set of transmitting devices and the set of clusters. Afterward, the local IDNC graph is extended to find the optimal file combination for a given cluster. Finally, the joint optimization problem is reformulated in terms of the non-interfering clusters and solved using the results of \sref{sec:int}. The complete proof can be found in \appref{app6}.
\end{proof}

In order to reduce the number of clusters, i.e., reduce $|\mathbf{Z}|$, reference \cite{9745154} proposes generating a subset of clusters. For their decoding delay setup, the authors propose sequentially constructing the extended cooperation graph by eliminating clusters that are surely not part of the maximum weight clique. The proposed method naturally extends to the completion time minimization by considering the adequate weight of vertices as defined in \thref{th4}.

\section{Simulation Results} \label{sec:sim}

\begin{figure}[t]
\centering
\includegraphics[width=0.8\linewidth]{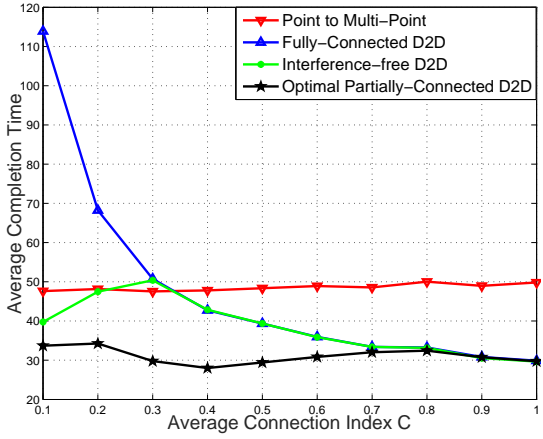}\\
\caption{Mean completion time versus the connectivity index $C$ for a network composed of $U=60$ devices, $F=30$ files, and an erasure probability $\mathbf{E}=0.1$.}\label{fig:C}
\end{figure}

\begin{figure}[t]
\centering
\includegraphics[width=0.8\linewidth]{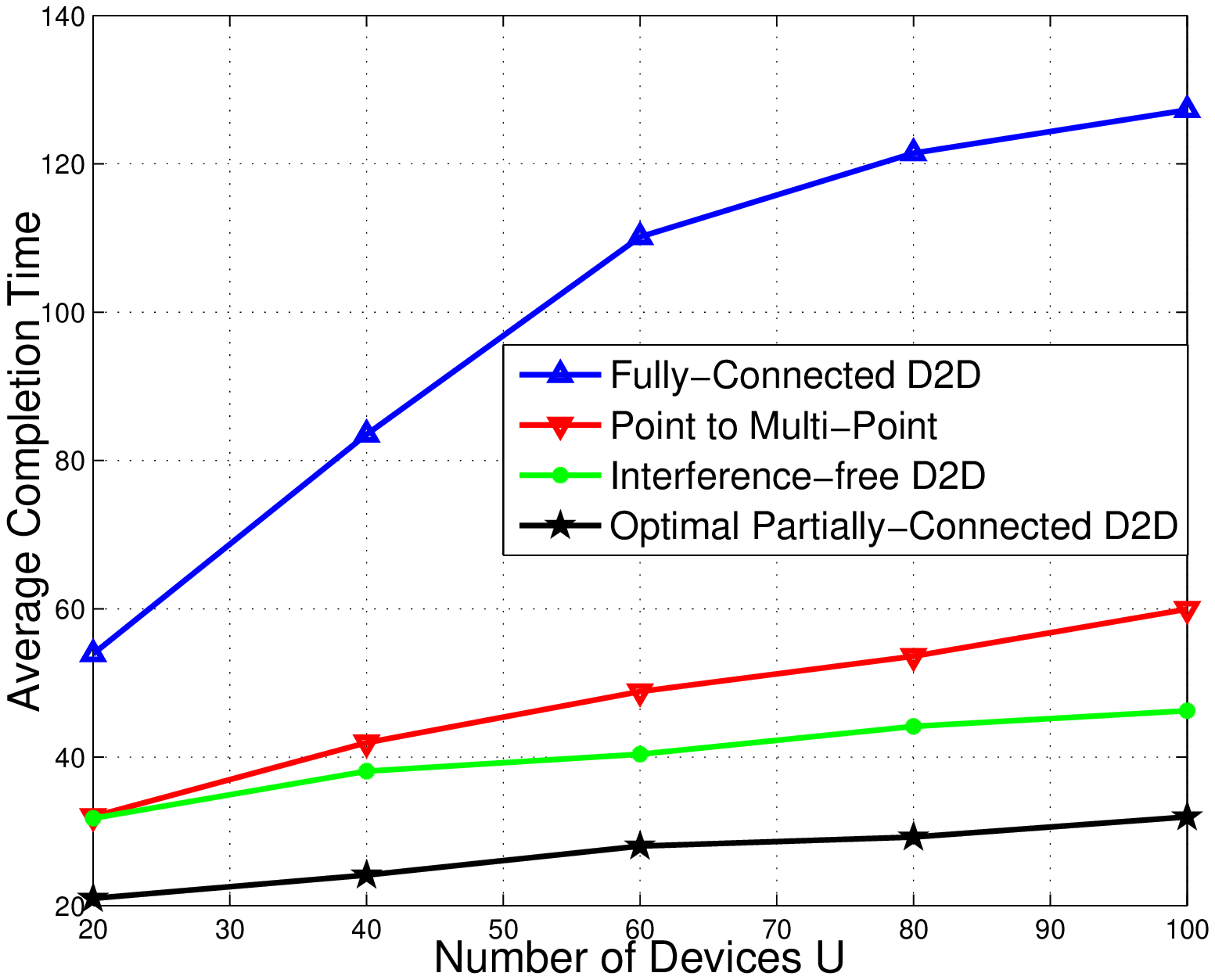}\\
\caption{Mean completion time versus number of devices $U$ for a network composed of $F=30$ files, a connecitvity index $C=0.1$, and an erasure probability $\mathbf{E}=0.1$.}\label{fig:M1}
\end{figure}

\begin{figure}[t]
\centering
\includegraphics[width=0.8\linewidth]{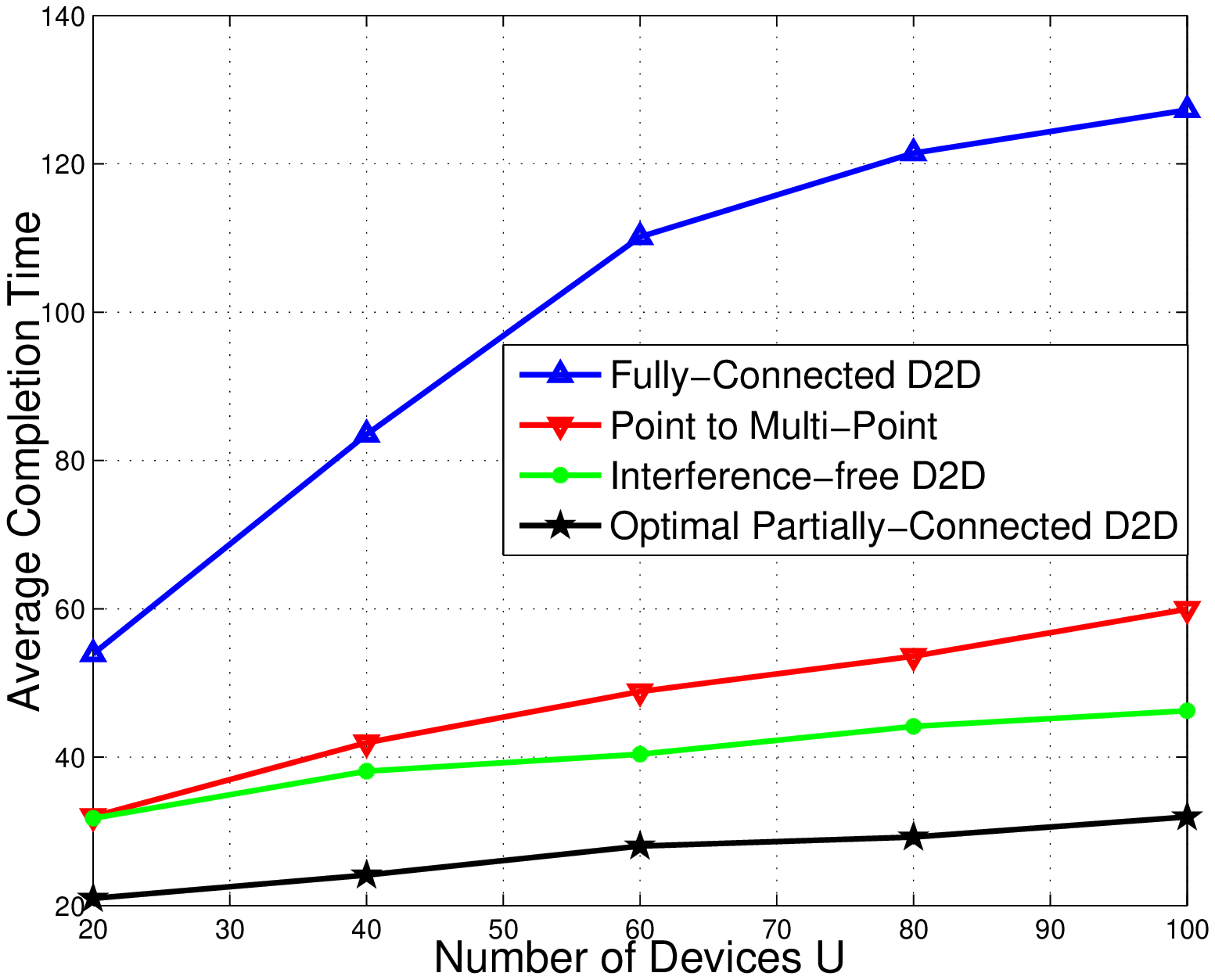}\\
\caption{Mean completion time versus number of devices $U$ for a network composed of $F=30$ files, a connecitvity index $C=0.4$, and an erasure probability $\mathbf{E}=0.1$.}\label{fig:M4}
\end{figure}

\begin{figure}[t]
\centering
\includegraphics[width=0.8\linewidth]{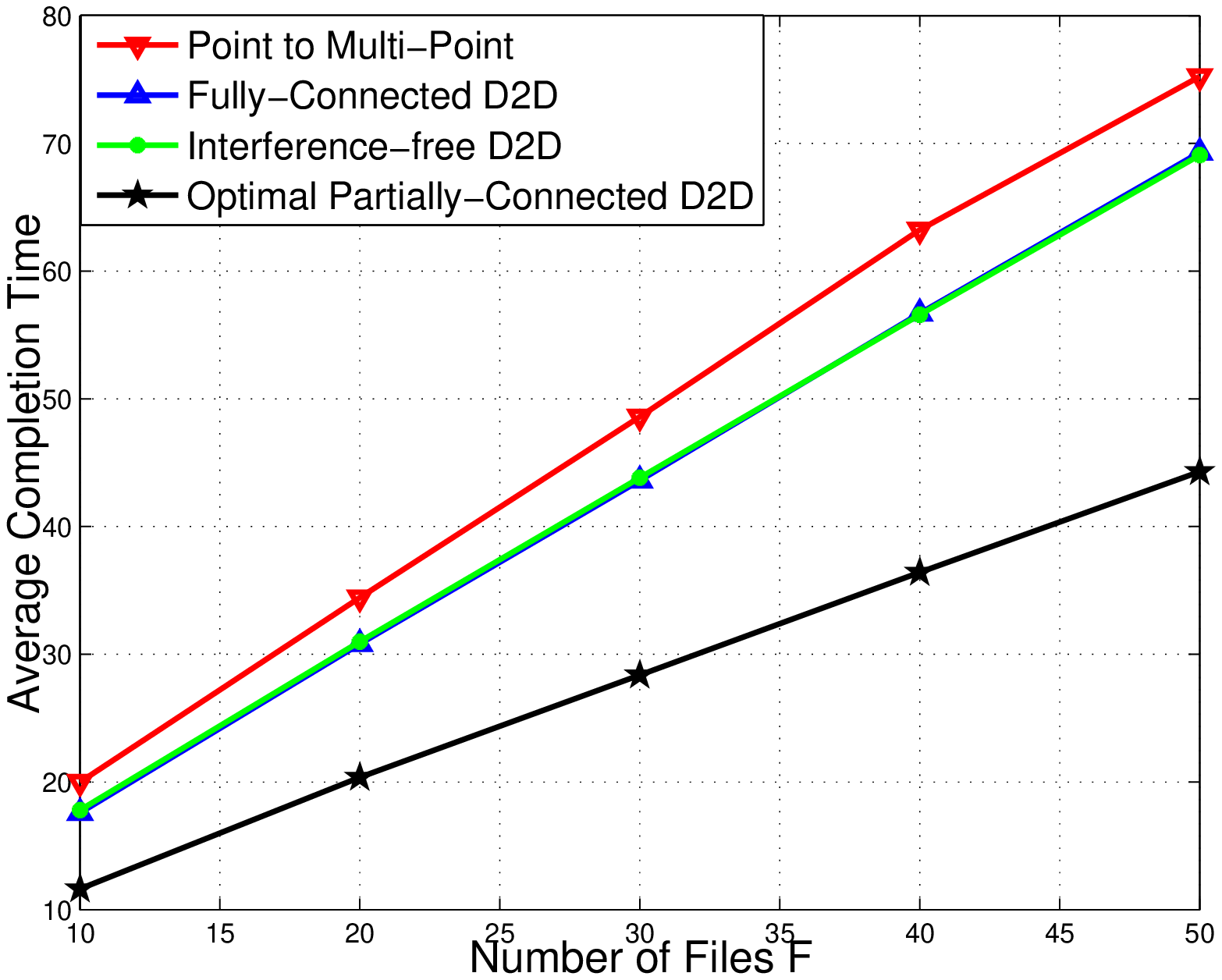}\\
\caption{Mean completion time versus number of files $F$ for a network composed of $U=60$ devices, a connecitvity index $C=0.1$, and an erasure probability $\mathbf{E}=0.1$.}\label{fig:N1}
\end{figure}

\begin{figure}[t]
\centering
\includegraphics[width=0.8\linewidth]{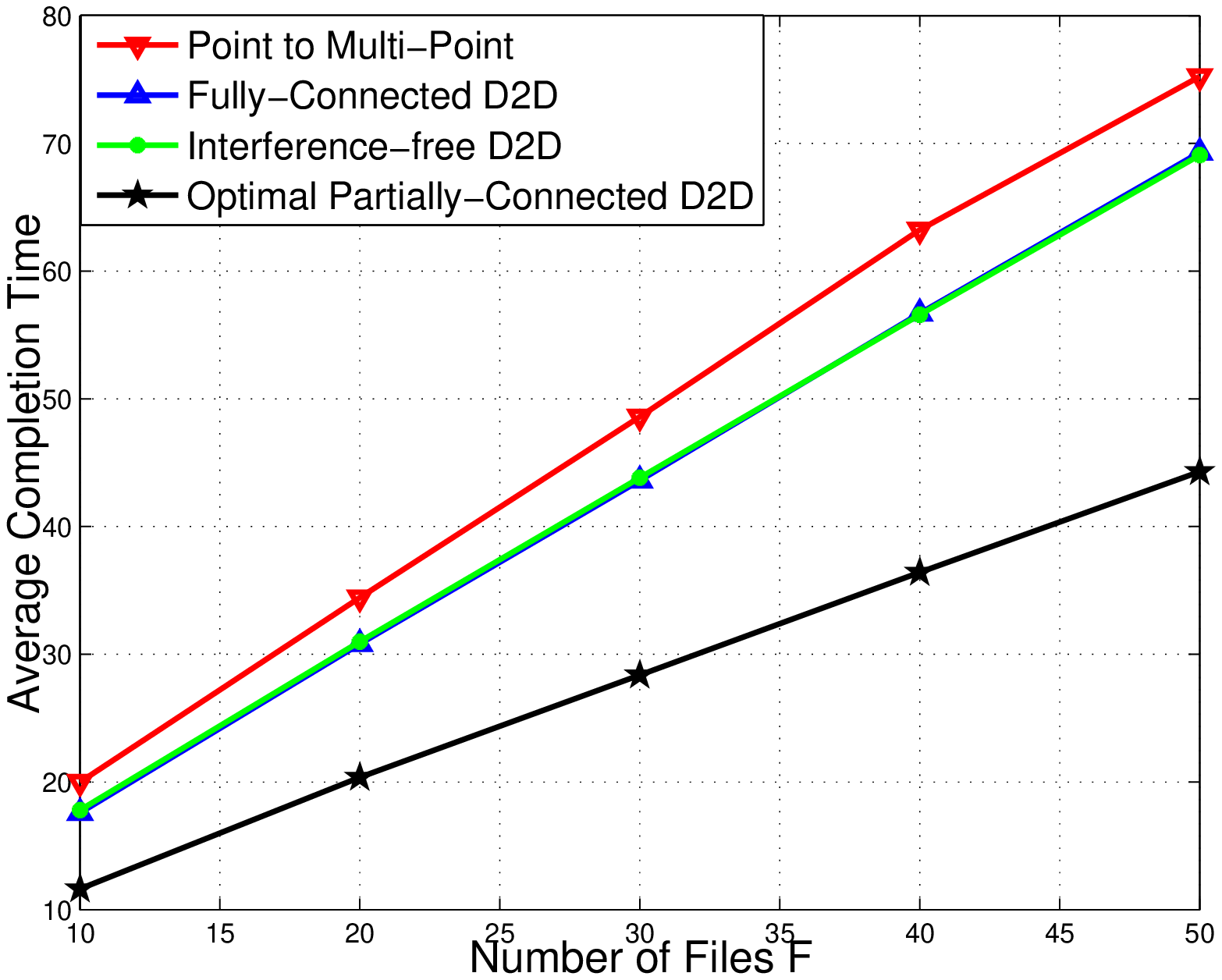}\\
\caption{Mean completion time versus number of files $F$ for a network composed of $U=60$ devices, a connecitvity index $C=0.4$, and an erasure probability $\mathbf{E}=0.1$.}\label{fig:N4}
\end{figure}

\begin{figure}[t]
\centering
\includegraphics[width=0.8\linewidth]{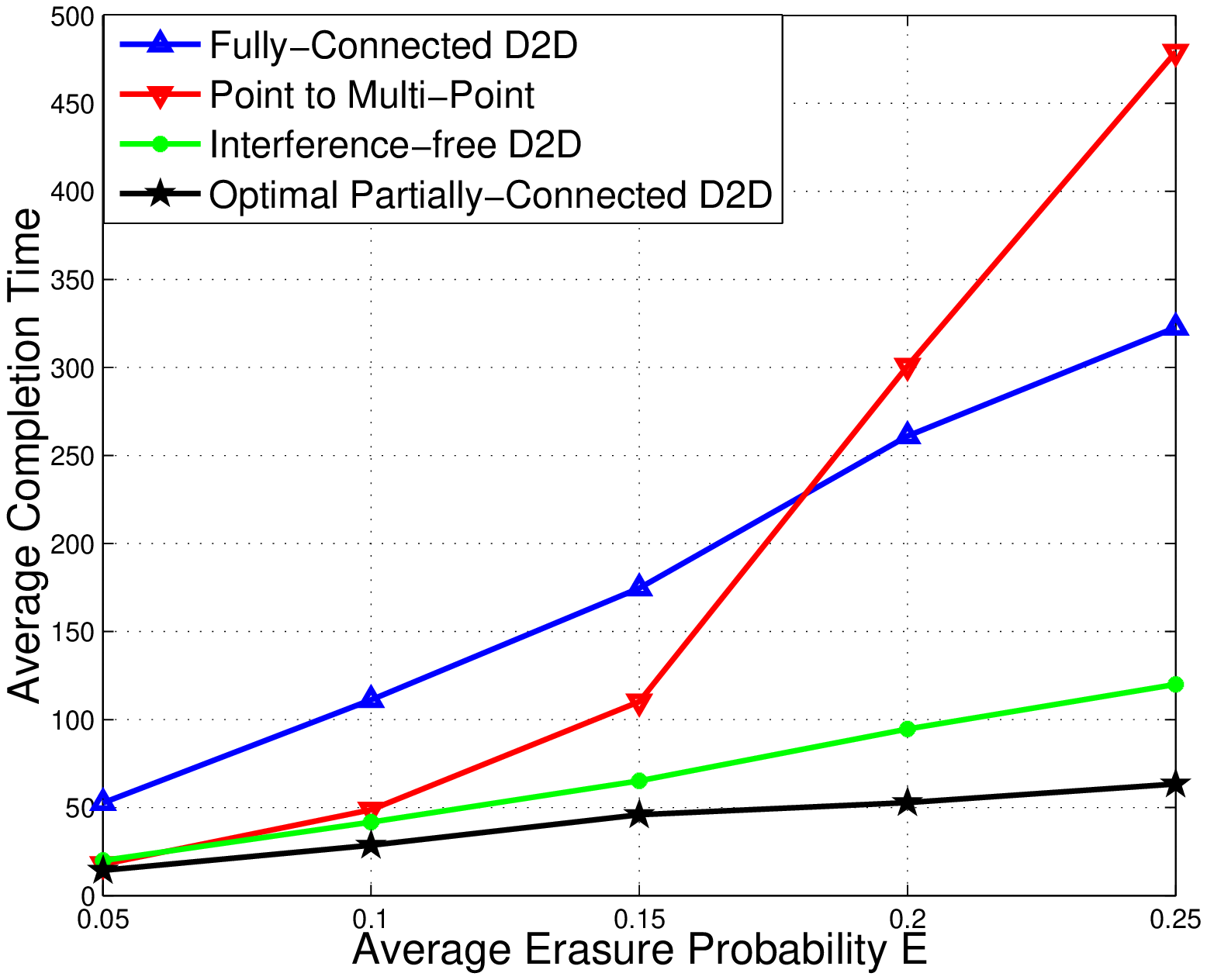}\\
\caption{Mean completion time versus average erasure probability $\mathbf{E}$ for a network composed of $U=60$ devices, $F=30$ files, and a connecitvity index $C=0.1$.}\label{fig:P1}
\end{figure}

\begin{figure}[t]
\centering
\includegraphics[width=0.8\linewidth]{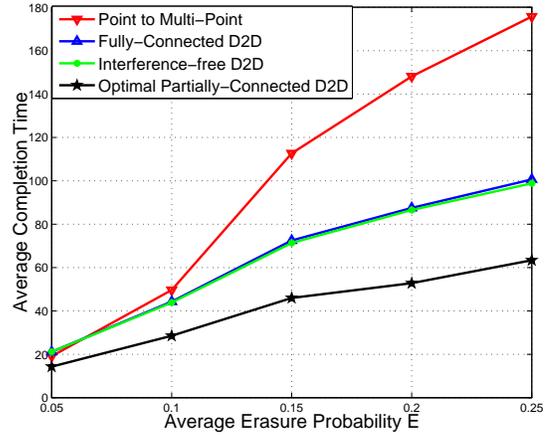}\\
\caption{Mean completion time versus average erasure probability $\mathbf{E}$ for a network composed of $U=60$ devices, $F=30$ files, and a connecitvity index $C=0.4$.}\label{fig:P4}
\end{figure}

This section attests the performance of the proposed interference-free and optimal solution to efficiently reduce the completion time in a partially connected D2D-enabled network. The proposed algorithms are compared against the following schemes:
\begin{itemize}
\item The PMP system in which a wireless base-station is responsible for the transmissions. The base-station can target all devices and hold all files. The average erasure probability from the base-station to the devices is denoted by $P$.
\item The fully-connected D2D system in which a single device transmits at each round. 
\end{itemize}

The completion time is computed over a large number of iterations and the average value presented on the plots. The initial distribution of the Has and Wants set of devices is drawn according to their average erasure probability. The connectivity index $C$ is defined as the number of edges normalized by the total number of edges, i.e., $C = |\mathcal{E}|/U^2$. The number of devices, files, erasure probability, and the connectivity index are variables in the simulations so as to study multiple scenarios. Given that the device-to-device channel is more reliable than the base-station-to-device one \cite{6620795,6404743}, then $P$ is fixed to $P=2 \mathbf{E}$ in all simulations.

\fref{fig:C} shows the average number of transmission against the connectivity index $C$ for a network composed of $U=60$ device, $F=30$ files and device-to-device erasure $\mathbf{E}=0.1$. The proposed interference-free solution provides a significant performance improvement as compared with the fully-connected algorithm for a poorly connected devices. This can be explained by the fact that for a poorly connected network, a probability of devices transmitting simultaneously while preserving the interference-free constraint is high. However, as the connectivity increases, both the interference-free and the fully connected solution provide the same performance as a single device is allowed to transmit. The optimal partially connected solution outperforms all other solutions for all connectivity index. However, for highly connected networks, all D2D solution have similar performance as collaboration between devices boils down to a single transmitting device.

\fref{fig:M1} and \fref{fig:M4} plot the completion time against the number of devices $U$ for a network composed of $F=30$ files, and an average device-to-device erasure $\mathbf{E}=0.1$ for a poorly connected (connectivity index $C=0.1$) and moderately connected (connectivity index $C=0.4$) networks, respectively. \fref{fig:N1} and \fref{fig:N4} illustrate the completion time against the number of files $F$ for a network composed of $U=60$ devices and an average device-to-device erasure $\mathbf{E}=0.1$  for a poorly connected (connectivity index $C=0.1$) and moderately connected (connectivity index $C=0.4$) networks, respectively.

For a low connectivity index in \fref{fig:M1} and \fref{fig:N1}, the interference-free solution provides performance gain over the point-to-multipoint solution even though the base-station can encode all combinations. This can be explained by the fact that the proposed solution allows multiple devices to transmit simultaneously. However, in the PMP solution and fully connected schemes only a single transmitting devices, and thus a single combination, can be communicated at each time slot. The optimal and interference-free solution have the same performance. This can be explained by the fact that the set of feasible combination of transmitting devices is large and thus allowing interference in poorly connected network does not provide significant performance improvement.

For moderately connected networks in \fref{fig:M4} and \fref{fig:N4}, the optimal solution outperforms both the fully connected and the interference-free solution. This is mainly due to the fact that the difference in size of the set of feasible transmitting devices in the interference-free and the optimal solution is no longer negligible. However for a small number of devices or files, the degradation of the interference-free solution against the optimal one is negligible as compared with the complexity gain.
 
\fref{fig:P1} and \fref{fig:P4} show the performance of the schemes to reduce the completion time versus the average device-to-device erasure $\mathbf{E}$ for a network composed of $U=60$ devices $F=30$ files for a poorly connected (connectivity index $C=0.1$) and moderately connected (connectivity index $C=0.4$) networks, respectively. The proposed optimal solution outperforms the conventional PMP algorithm for all values of the erasure probability. This can be explained by the fact that the optimal solution allows multiple transmissions and thus it allows satisfying multiple devices with non-combinable file demand. 

\section{Conclusion} \label{sec:con}

This paper considers the completion time reduction problem in a partially connected device-to-device network using instantly decodable network coding. The joint problem over the set of transmitting devices and the file combinations is formulated and solved. The proposed solution relies on finding the file combination using the local IDNC graph and using it to construct the cooperation graph and solve a maximum weight clique problem. Simulation results show that the proposed solution largely outperform conventional approaches. As a future research direction, the fully distributed system in which decision are made locally at each device could be considered. Another interesting direction is the multicast scenario in which the demand of each device differs from the others. In that case, a wanted file by one device may not unwanted by all its neighbors.

\appendices
\numberwithin{equation}{section}

\section{Proof of \lref{l1}} \label{app1}

Let $\mathcal{S}$ be a schedule of transmissions. The individual completion time $\mathcal{T}_u(\mathcal{S})$ of the $u$-th device occurs when the device receives an instantly decodable file combination that makes its Wants set empty. Such event occurs in the $\mathcal{T}_u(\mathcal{S})$-th recovery transmission. For time slots $t$ before the $\mathcal{T}_u(\mathcal{S})$-th transmission, the following events can occur at the $u$-th device:
\begin{itemize}
\item No file can be heard by the device. Such event happens if one of the following scenarios happens:
\begin{itemize}
\item The $u$-th device is one of the transmitting devices.
\item The $u$-th device experiences interference from the transmitting devices.
\item The $u$-th device is out of the transmission range of the transmitting devices.
\end{itemize}
In all above cases, the $u$-th device is not able to hear exactly a single transmission. Therefore, its cumulative decoding delay $\mathcal{D}_u(\mathcal{S})$ increases by one unit.
\item A single file combination can be heard by the device. Two events can occur:
\begin{itemize}
\item The file combination is erased at the $u$-th device. Therefore, the number of erased files $\mathcal{E}_u(\mathcal{S})$ increases by a single unit.
\item The file combination is successfully received by the device. Depending on the instant decodability of the combination, two events can occur:
\begin{itemize}
\item The combination is instantly decodable. The device needs to receive $|\mathcal{W}_u|-1$ of such combinations before the $\mathcal{T}_u(\mathcal{S})$-th time slot.
\item The combination is not instantly decodable. From its definition, the cumulative decoding delay $\mathcal{D}_u(\mathcal{S})$ of the $u$-th devices increases by one unit.
\end{itemize}
\end{itemize}
\end{itemize}

In the $\mathcal{T}_u(\mathcal{S})$-th transmission, the $u$-th device completes the reception of all its Wanted files. Therefore, number of recovery transmission $\mathcal{T}_u(\mathcal{S})$ of the $u$-th device can be expressed using the following formula: 
\begin{align}\label{eq2}
\mathcal{T}_u(\mathcal{S}) = |\mathcal{W}_u| + \mathcal{D}_u(\mathcal{S}) + \mathcal{E}_u(\mathcal{T}_u(\mathcal{S})-1).
\end{align}
Finally, note that the $\mathcal{T}_u(\mathcal{S})$-th transmission results in a successful transmission of the $u$-th device as it receives its missing file. Therefore, the number of erased files is $\mathcal{E}_u(\mathcal{T}_u(\mathcal{S})-1) = \mathcal{E}_u(\mathcal{T}_u(\mathcal{S}))=\mathcal{E}_u(\mathcal{S})$ which allows to reformulate the expression \eref{eq2} as:
\begin{align}
\mathcal{T}_u(\mathcal{S}) = |\mathcal{W}_u| + \mathcal{D}_u(\mathcal{S}) + \mathcal{E}_u(\mathcal{S}).
\end{align}

\section{Proof of \thref{th1}}\label{app2}

This theorem is shown by approximating the number of erased files at the $u$-th device. The assumption that devices are equally likely to transmit follows from the fact that the Has and Wants set are randomly and uniformly distributed among devices. Further, for a moderately connected networks, the number of transmissions in which the $u$-th device is out of the transmission range of the transmitting devices $\mathcal{A}$ is negligible with respect to the total number of transmissions, i.e., $\mathcal{J} = \varnothing$.

Let $\mathcal{C}_u$ be the coverage zone of $u$-th device. The following Lemma characterizes the transmission probability of devices in the coverage zone $\mathcal{C}_u$:
\begin{lemma}
At each time instant $t$ one and exactly one device in $\mathcal{C}_u$ transmits. Moreover, devices in $\mathcal{C}_i$ are all equally likely to transmit.
\label{l2}
\end{lemma}

\begin{proof}
The lemma is proven by first showing that at least a single device $\mathcal{C}_u$ transmits at each time instant. Secondly, it is shown that, under the interference-free assumption $\mathcal{I} = \varnothing$, no more than a single device is allowed to transmit. Finally, extending the equally likely transmission property of all devices to the coverage zone $\mathcal{C}_u$ concludes the proof. The complete proof can be found in \appref{app5}.
\end{proof}

Let $\mathcal{X}_u(t)$ be Bernoulli random variable that takes value $1$ if the file is erased at the $u$-th device in the $t$-th transmission and $0$ otherwise and let $\mathcal{L}(t)$ be a random variable taking the index of the transmitting device $u^\prime \in \mathcal{C}_i$. The probability for a file to be erased at the $u$-th device in the $t$-th transmission can be expressed as:
\begin{align}
\mathds{P}(\mathcal{X}_u(t)=1) = \sum_{u^\prime \in \mathcal{C}_u}\mathds{P}(\mathcal{X}_u(t)=1|\mathcal{L}(t)=u^\prime)\mathds{P}(\mathcal{L}(t)=u^\prime).
\end{align}

We have $\mathds{P}(\mathcal{X}_u(t)=1|\mathcal{L}(t)=u)= 0$ since the transmitting device is the $u$-th device itself and thus the file combination cannot be erased. From the system model, the erasure probability is given by
\begin{align}
\mathds{P}(\mathcal{X}_u(t)=1|\mathcal{L}(t)=u^\prime) = \epsilon_{u^\prime u}.
\end{align} 

From \lref{l1}, all devices in $\mathcal{C}_u$ are equally likely to transmit. Therefore:
\begin{align}
\mathds{P}(\mathcal{L}(t)=u^\prime) = \cfrac{1}{|\mathcal{C}_u|}, \forall \ u^\prime \in \mathcal{C}_u.
\end{align}
Combining the two previous equation, the probability that the file combination is erased at the $u$-th device can be expressed as :
\begin{align}
\mathds{P}(\mathcal{X}_u(t)=1) = \cfrac{|\mathcal{C}_u|-1}{|\mathcal{C}_u|} \sum_{u^\prime \neq u \in \mathcal{C}_u}\epsilon_{u^\prime u} = \cfrac{1}{|\mathcal{C}_u|} \overline{\epsilon}_{u}
\end{align}
where $\overline{\epsilon}_{u} = \cfrac{1}{|\mathcal{C}_u|-1} \sum_{u^\prime \neq u \in \mathcal{C}_u}\epsilon_{u^\prime u}$ is the average file erasure probability experienced by the $u$-device.

The cumulative number of erased files at the $t$-th device until time slot $n$ can be written as:
\begin{align}
\mathcal{E}_u(n) &= \sum_{t=1}^{n} \mathcal{X}_u(t)
\end{align}
Hence, the total number of erased files at the $t$-th device until its completion time $\mathcal{E}_u(\mathcal{T}_u(\mathcal{S}))$ is the sum of $\mathcal{C}_u(\mathcal{S})-1$ Bernoulli variable $\mathcal{X}_u(t)$. For large enough number of files $F$, the individual completion time $\mathcal{T}_u(\mathcal{S})$ of the $t$-th device would also be large enough. Therefore, we can the law of large numbers to approximate $\mathcal{E}_u(\mathcal{T}_u(\mathcal{S})-1)$ as follows:
\begin{align}\label{eq1}
\mathcal{E}_u(\mathcal{T}_u(\mathcal{S})-1) \approx \overline{\epsilon}_{u} \cfrac{|\mathcal{C}_u|-1}{|\mathcal{C}_u|}(\mathcal{T}_u(\mathcal{S})-1).
\end{align}
Replacing \eqref{eq1} in the expression provided in \lref{l1} and re-arranging the terms, the completion time for the $t$-th device can be finally expressed as:
\begin{align}
\mathcal{T}_u(\mathcal{S}) = \cfrac{|\mathcal{W}_u| + \mathcal{D}_u(\mathcal{S}) - \overline{\epsilon}_u\cfrac{|\mathcal{C}_u|-1}{|\mathcal{C}_u|}}{1-\overline{\epsilon}_u\cfrac{|\mathcal{C}_u|-1}{|\mathcal{C}_u|}}.
\end{align}
Thus, the expression for the overall completion time can be expressed as:
\begin{align}
\mathcal{C}(S) \approx \max_{i\in\mathcal{M}}\left\{ \mathcal{T}_u(\mathcal{S}) = \cfrac{|\mathcal{W}_u| + \mathcal{D}_u(\mathcal{S}) - \overline{\epsilon}_u\cfrac{|\mathcal{C}_u|-1}{|\mathcal{C}_u|}}{1-\overline{\epsilon}_u\cfrac{|\mathcal{C}_u|-1}{|\mathcal{C}_u|}} \right\}
\end{align}

\section{Proof of \thref{th2}}\label{app3}

The proof of this theorem first formulates the online completion time reduction problem as a joint optimization over the set of transmitting devices and their file combination. Afterward, using the definition of the critical set and the network topology, the problem can be reformulated as a constrained optimization wherein the objective function represents the set of transmitting devices and the constraint the file combinations. Finally, using the expression of the decoding delay provided in \cite{5456269}, the optimization problem is explicitly formulated.

The online optimization problem that reduce the probability of increase in the anticipated completion time is given in the following lemma:
\begin{lemma}
Let $\mathcal{A}$ be the set of transmitting devices and $\kappa_a(\mathcal{A}), a \in \mathcal{A}$ the file combination of the $a$-th device. The joint online optimization problem that reduces the completion time can be formulated as:
\begin{align}
\max_{\substack{ \mathcal{A} \in \mathcal{P}(\mathcal{U}) \\ \kappa_a(\mathcal{A}) \in \mathcal{P}(\mathcal{H}_{a})}} \left\{\prod_{u \in \mathcal{L}} \mathds{P} \left[d_u(\mathcal{A},\kappa_{a}(\mathcal{A}))=0\right]\right\}.
\end{align}
where $d_u(\mathcal{A},\kappa_{a}(\mathcal{A}))$ refers to the decoding delay increase of the $u$-th device when devices $a \in \mathcal{A}$ are transmitting the combination $\kappa_a(\mathcal{A})$.
\label{l3}
\end{lemma}

\begin{proof}
The lemma is shown by expressing the probability of the anticipated completion time increase . The joint optimization over the set of transmitting devices and the file combinations is formulated as a minimization of the probability increase in the anticipated completion time. The complete proof can be found in \appref{app5}.
\end{proof}

Let $\mathcal{O}_u(\mathcal{A})$ be the opportunity zone of the $u$-th device defined such that devices in that zone can hear the transmission from the $u$-th device and decode a file from it. The mathematical definition of this opportunity zone is given by the equation below:
\begin{align}
\mathcal{O}_u(\mathcal{A}) = \tilde{\mathcal{U}} \cap \left( \mathcal{C}_u \setminus \left( \mathcal{A} \cup \mathcal{I} \right) \right)
\end{align}

Let $\tau_u(\kappa_u(\mathcal{A}))$ be the set of targeted devices by the $u$-th device when sending the file combination $\kappa_u$ and devices in $\mathcal{A}$ are transmitting. From \cite{5456269}, the distribution of the decoding delay $d_u(\mathcal{A},\kappa_{a}(\mathcal{A}))$ is given by:
\begin{align}\label{eq6}
\mathds{P}[d_u(\mathcal{A},\kappa_{a}(\mathcal{A})) = 0] =
\begin{cases}
1  &\text{ if } u \in \mathcal{U} \setminus \tilde{\mathcal{U}}\\
0  &\text{ if } u \in (\mathcal{A} \cup \mathcal{I} \cup \mathcal{J}) \cap \tilde{\mathcal{U}} \\
1  &\text{ if } u \in \mathcal{O}_{a}(\mathcal{A})\cap \tau_{a}(\kappa_{a}(\mathcal{A})) \\
\epsilon_{au} &\text{ if } i \in \mathcal{O}_{a}(\mathcal{A})  \setminus \tau_{a}(\kappa_{a}(\mathcal{A}))
\end{cases}
\end{align}

Using the expressions of the decoding delay increment in \eref{eq6} the probability that all the devices in the critical set $\mathcal{L}$ do not experience a decoding delay can be expressed as:
\begin{align}
&\prod_{u \in \mathcal{L}} \mathds{P} \left[d_u(\mathcal{A},\kappa_{a}(\mathcal{A}))=0\right] = \nonumber \\
&\begin{cases} 
 0 &\text{ if } \mathcal{A} \notin \mathcal{R}(\mathcal{L}) \\
\prod\limits_{a \in \mathcal{A}} \prod\limits_{u \in \mathcal{L} \cap (\mathcal{O}_{a}(\mathcal{A}) \setminus \tau_{a}(\kappa_{a}(\mathcal{A})))} \epsilon_{au} &\text{ if } \mathcal{A} \in \mathcal{R}(\mathcal{L})  
\end{cases},
\end{align}
where the set $\mathcal{R}(\mathcal{L}) = \{ \mathcal{A} \in \mathcal{P}(\mathcal{U}) | \mathcal{L} \cap \tilde{\mathcal{U}} \cap (\mathcal{A} \cup \mathcal{I} \cup \mathcal{J} ) = \varnothing \}$ represents the set of feasible combination of devices. Clearly, a combination of devices that is not included in the feasible set increase the expected completion time with probability $1$ in the considered transmission. Therefore, the problem of finding the optimal set of transmitting devices and their optimal file combinations can be written as:
\begin{align}
&\max_{\substack{ \mathcal{A} \in \mathcal{P}(\mathcal{U}) \\ \kappa_a(\mathcal{A}) \in \mathcal{P}(\mathcal{H}_{a})}} \left\{\prod_{u \in \mathcal{L}} \mathds{P} \left[d_u(\mathcal{A},\kappa_{a}(\mathcal{A}))=0\right]\right\} \nonumber \\
& \max_{\substack{ \mathcal{A} \in \mathcal{R}(\mathcal{L}) \\ \kappa_a(\mathcal{A}) \in \mathcal{P}(\mathcal{H}_{a})}} \left\{\prod\limits_{a \in \mathcal{A}} \prod\limits_{u \in \mathcal{L} \cap (\mathcal{O}_{a}(\mathcal{A}) \setminus \tau_{a}(\kappa_{a}(\mathcal{A})))} \epsilon_{au} \right\} \nonumber \\
& \max_{\substack{ \mathcal{A} \in \mathcal{R}(\mathcal{L}) \\ \kappa_a(\mathcal{A}) \in \mathcal{P}(\mathcal{H}_{a})}} \left\{\sum\limits_{a \in \mathcal{A}} \sum\limits_{u \in \mathcal{L} \cap (\mathcal{O}_{a}(\mathcal{A}) \setminus \tau_{a}(\kappa_{a}(\mathcal{A})))} \log(\epsilon_{au}) \right\} \nonumber \\
& \min_{\substack{ \mathcal{A} \in \mathcal{R}(\mathcal{L}) \\ \kappa_a(\mathcal{A}) \in \mathcal{P}(\mathcal{H}_{a})}} \left\{\sum\limits_{a \in \mathcal{A}} \sum\limits_{u \in \mathcal{L} \cap \tau_{a}(\kappa_{a}(\mathcal{A}))} \log(\epsilon_{au}) \right\} \nonumber \\
& \max_{\substack{ \mathcal{A} \in \mathcal{R}(\mathcal{L}) \\ \kappa_a(\mathcal{A}) \in \mathcal{P}(\mathcal{H}_{a})}} \left\{\sum\limits_{a \in \mathcal{A}} \sum\limits_{u \in \mathcal{L} \cap \tau_{a}(\kappa_{a}(\mathcal{A}))} \log\left(\cfrac{1}{\epsilon_{au}}\right) \right\}
\end{align}

Let $y(\kappa_{a}(\mathcal{A})) = \sum\limits_{u \in \mathcal{L} \cap \tau_{a}(\kappa_{a}(\mathcal{A}))} \log \cfrac{1}{\epsilon_{au}}$. Therefore, the problem of finding the optimal set of transmitting devices $\mathcal{A}$ and their file combination $\kappa_{a}(\mathcal{A}), \ \forall\ a \in \mathcal{A}$ can be expressed as a constrained optimization as follows:
\begin{subequations}
\begin{align}
\mathcal{A}^* &= \arg \max_{ \mathcal{A} \in \mathcal{R}(\mathcal{L}) } \sum\limits_{a \in \mathcal{A}} y(\kappa_{a}^*(\mathcal{A})) \\
\text{subject to } \kappa_{a}^*(\mathcal{A}) &= \arg \max_{\kappa_{a}(\mathcal{A}) \in \mathcal{P}(\mathcal{H}_{a})} y(\kappa_{a}(\mathcal{A}))  \\
y(\kappa_{a}(\mathcal{A})) &= \sum\limits_{u \in \mathcal{L} \cap \tau_{a}(\kappa_{a}(\mathcal{A}))} \log \cfrac{1}{\epsilon_{au}}  
\end{align} 
\end{subequations}

\section{Proof of \thref{th3}}\label{app4}

Under the interference-free transmissions constraint, the authors in \cite{5456269} show that the optimization over the set of transmitting devices and their file combination can be separated. In other words, the optimization problem is written as:
\begin{subequations}
\label{eq:th3} 
\begin{align}
\mathcal{A}^* = \arg &\max_{ \mathcal{A} \in \mathcal{R}(\mathcal{L}) \cap \mathcal{N} } \sum\limits_{a \in \mathcal{A}} \sum\limits_{u \in \mathcal{L} \cap \tau_{a}(\kappa_{a}^* )} \log \cfrac{1}{\epsilon_{au}} \label{eq:th2} \\
\text{subject to } \kappa_{a}^* &= \arg \max_{\kappa_{a} \in \mathcal{P}(\mathcal{H}_{a})} \sum\limits_{u \in \mathcal{L} \cap \tau_{a}(\kappa_{a})} \log \cfrac{1}{\epsilon_{au}} \label{eq:th1}
\end{align} 
\end{subequations}

The formulation in \eref{eq:th1} shows that the file combination does not depend on the set of transmitting devices and thus it can be solved for all devices. Following similar steps than the one used in \cite{13051412}, the optimal file combination that the $u$-th device can generate is obtained by solving a maximum weight clique problem over the multi-layer IDNC graph. However, while reference \cite{13051412} considers that the $u$-th devices can target all other devices, devices represented by vertices in the proposed multi-layer are those in the transmission range of the $u$-th device. Furthermore, according to the expression of the anticipated completion time in \thref{th1}, the construction of the $n$-th layer is given by:
\begin{itemize}
\item $\mathcal{T}_{u^\prime}(t-1)+\cfrac{n}{1-\overline{\epsilon}_{u^\prime}\cfrac{|\mathcal{C}_{u^\prime}|-1}{|\mathcal{C}_{u^\prime}|}} > \max\limits_{u^{\prime \prime} \in \tilde{\mathcal{U}}} \left(\mathcal{T}_{u^{\prime \prime}}(t-1)\right)$
\item $\mathcal{T}_{u^\prime}(t-1)+\cfrac{n-1}{1-\overline{\epsilon}_{u^\prime}\cfrac{|\mathcal{C}_{u^\prime}|-1}{|\mathcal{C}_{u^\prime}|}} \leq \max\limits_{u^{\prime \prime} \in \tilde{\mathcal{U}}} \left(\mathcal{T}_{u^{\prime \prime}}(t-1)\right)$
\end{itemize}

Finally, the weight of vertex $v_{{u^\prime}f}$ in the multi-layer IDNC graph is given by:
\begin{align}
w(v_{u^\prime f}) = - \log (\epsilon_{u u^\prime}).
\end{align}

After computing the optimal file combination $\kappa_{u}^*$, let $y^*_{u} = \sum\limits_{u^\prime \in \mathcal{L} \cap \tau_{u}(\kappa_{u}^* )} -\log {\epsilon_{u u^\prime}}$ be the contribution of the $u$-th device to the network. The problem of selecting the set of transmitting devices becomes:
\begin{align}
\max_{ \mathcal{A} \in \mathcal{R}(\mathcal{L}) \cap \mathcal{N} } \sum\limits_{u \in \mathcal{A}} y^*_{u}.
\label{eq:th4} 
\end{align} 

Clearly, each solution to the above problem represents a clique in the cooperation graph. Indeed assume a couple of vertices is not connected, then the resulting set $\mathcal{A}$ violate the interference-free constraint. Similarly, vertices $v_u$ that represent a clique in the graph are a valid solution to the optimization \eref{eq:th4}. Finally, the weight of each clique coincides with the objective function of its corresponding set $\mathcal{A}$. Therefore, the optimal solution is the maximum weight clique in the cooperation graph wherein the weight of each vertex $v_{u}$ is:
\begin{align}
w(v_u) = y^*_{u} = \sum\limits_{u^\prime \in \mathcal{L} \cap \tau_{u}(\kappa_{u})} \log \cfrac{1}{\epsilon_{uu^\prime}},
\end{align}

\section{Proof of \thref{th4}}\label{app6}

The theorem is established by showing a one-to-one mapping between all feasible set of transmitting devices and the set of clusters. Afterward, the local IDNC graph is extended to find the optimal file combination for a given cluster. Finally, the joint optimization problem is reformulated in terms of the non-interfering clusters and solved using the results of \sref{sec:int}.

The following preposition shows there is a bijection between the set of transmitting devices and a set of clusters verifying certain properties:
\begin{preposition}
For each set of transmitting devices $\mathcal{A}$, there exists a unique set $\mathfrak{Z} \subset \mathbf{Z}$ satisfying the following constraints:
\begin{subequations}
\begin{align}
\bigoplus_{\mathcal{Z} \in \mathfrak{Z}} \mathcal{Z} &= \mathcal{A} \label{eq:th13} \\ 
\mathcal{C}^T (\mathcal{Z}) \cap \mathcal{C}^T (\mathcal{Z}^\prime) &= \varnothing,\ \forall\ \mathcal{Z} \neq \mathcal{Z}^\prime \in \mathfrak{Z} \label{eq:th12} \\
\mathcal{C}^T(Z) \cap \mathcal{C}^T (\mathcal{Z} \setminus Z ) &\neq \varnothing,\ \forall \ Z \subset \mathcal{Z}, \mathcal{Z} \in \mathfrak{Z},  \label{eq:th11}
\end{align}
\end{subequations}
\end{preposition}

\begin{proof}
The proof of this preposition is omitted as it mirrors the steps used in proving Lemma 2 in \cite{9745154}.
\end{proof}

Using the preposition above, showing the one-to-one mapping between the set of feasible transmitting devices and the set of cliques in the extended cooperation graph boils down to proving that the corresponding $\mathfrak{Z}$ is a clique.

Let $\mathcal{A} \in \mathcal{N} \cap \mathcal{R}(\mathcal{L})$ be a feasible set of transmitting devices. Given that all clusters $\mathcal{Z} \in \mathcal{N} \cap \mathcal{R}(\mathcal{L})$ verifying \eref{eq:th11} are generated, then there exists a subset of $\mathbf{Z}$ verifying \eref{eq:th13}. Furthermore, given that the connectivity condition of the extended cooperation graph matches the constraint \eref{eq:th12}, then all vertices in $\mathbf{Z}$ are connected. Finally, we conclude there is a one-to-one mapping between the set of feasible transmitting devices and the cliques $\mathbf{Z}$ in the extended cooperation graph.

Using the bijection above, the completion time joint optimization problem can be reformulated as follows:
\begin{subequations}
\begin{align}
&\max_{\mathfrak{Z}} \sum_{\mathcal{Z} \in \mathfrak{Z}} \sum_{u \in \mathcal{Z}} \sum_{u^\prime \in \tau_u(\kappa_u^*(\mathfrak{Z}))} \cfrac{1}{\log(\epsilon_{uu^\prime})} \label{eqth14} \\
\text{s.t. } \kappa_u^*(\mathfrak{Z}) &= \arg \max_{\kappa_{u}(\mathfrak{Z}) \in \mathcal{P}(\mathcal{H}_{u})} \sum_{u^\prime \in \tau_u(\kappa_u^*(\mathfrak{Z}))} \cfrac{1}{\log(\epsilon_{uu^\prime})} \label{eqth15}
\end{align}
\end{subequations}

Due to the interference-free cluster generation given in \eref{eq:th12}, the file combination only depends on its own cluster, i.e., $\kappa_u^*(\mathfrak{Z}) = \kappa_u^*(\mathcal{Z})$. Such property allows the separation of both problems as follows:
\begin{subequations}
\begin{align}
&\max_{\mathfrak{Z}} \sum_{\mathcal{Z} \in \mathfrak{Z}} \sum_{u \in \mathcal{Z}} \sum_{u^\prime \in \tau_u(\kappa_u^*(\mathcal{Z}))} \cfrac{1}{\log(\epsilon_{uu^\prime})} \label{eqth16} \\
\text{s.t. } \kappa_u^*(\mathfrak{Z}) &= \arg \max_{\kappa_{u}(\mathcal{Z}) \in \mathcal{P}(\mathcal{H}_{u})} \sum_{u^\prime \in \tau_u(\kappa_u^*(\mathfrak{Z}))} \cfrac{1}{\log(\epsilon_{uu^\prime})} \label{eqth17}
\end{align}
\end{subequations}

The difference between the new file combination optimization problem \eref{eqth17} and the one proposed in \sref{sec:int} is that some devices are transmitting and some are in interference. Therefore, the optimal file combination is obtained by solving the maximum weight clique in the extended multi-layer graph that excludes those devices. Finally, using the results of \thref{th3}, the optimal solution to the joint optimization problem \eref{eq:4} is equivalent to a maximum weight clique in the extended cooperative graph wherein the weight vertex $v$ corresponding to cluster $\mathcal{Z}$ is given by:
\begin{align}
w(v) = \sum_{u \in \mathcal{Z}} \sum_{u^\prime \in \tau_u(\kappa_u(\mathcal{Z}))} \cfrac{1}{\log(\epsilon_{uu^\prime})}
\end{align}
and $\kappa_u(\mathcal{Z})$ is obtained by solving the maximum weight clique problem in the extended multi-layer IDNC graph $\mathcal{G}_u(\mathcal{Z})$ wherein the weight of vertex $v_{u^\prime f}$ is:
\begin{align}
w(v_{u^\prime f}) = - \log (\epsilon_{u u^\prime}).
\end{align}

\section{Proof of Auxiliary Lemmas}\label{app5}

\subsection{Proof of \lref{l2}} 

The lemma is proven by first showing that at least a single device $\mathcal{C}_u$ transmits at each time instant. Secondly, it is shown that, under the interference-free assumption $\mathcal{I} = \varnothing$, no more than a single device is allowed to transmit. Finally, extending the equally likely transmission property of all devices to the coverage zone $\mathcal{C}_u$ concludes the proof.

Assume that for a given time slot $t$, none of the devices in $\mathcal{C}_u$ is transmitting, i.e., $\mathcal{A} \cap \mathcal{C}_u = \varnothing$. From the symmetry of the connectivity matrix, the following holds:
\begin{align}
\forall \ u^\prime \notin \mathcal{C}_u, u \notin \mathcal{C}_{u^\prime}, \Rightarrow u \notin \bigcup_{{u^\prime} \in \mathcal{U} \setminus \mathcal{C}_u} \mathcal{C}_{u^\prime}
\end{align}
The total coverage zone of the transmitting devices is $\mathcal{C}^T(\mathcal{A})$ defined by:
\begin{align}
\mathcal{C}^T(\mathcal{A}) = \bigcup\limits_{{u^\prime} \in \mathcal{A}} \mathcal{C}_{u^\prime} \subseteq \bigcup\limits_{{u^\prime} \in \mathcal{U} \setminus \mathcal{C}_u} \mathcal{C}_{u^\prime}.
\end{align}
By definition of the set $\mathcal{J}$, we have:
\begin{align}
u \notin \bigcup_{{u^\prime} \in \mathcal{U} \setminus \mathcal{C}_u} \mathcal{C}_{u^\prime} \Rightarrow u \notin \mathcal{C}^T(A) \Rightarrow i \in \mathcal{J} \Rightarrow \mathcal{J} \neq \varnothing.
\end{align}

However, by assumption, transmission in which the $u$-th device is out of the transmission range of the transmitting devices are negligible. Therefore, at least a single device $\mathcal{C}_u$ transmits at each time instant.

Now assume that at least $2$ devices $u_1$ and $u_2$ from $\mathcal{C}_u$ transmit simultaneously. By definition of the interference region $\mathcal{I}$ and the symmetry of the connectivity matrix, the following hold:
\begin{align}
u \in \mathcal{C}_{u_1} \cap \mathcal{C}_{u_2} \text{ with } (u_1,u_2) \in \mathcal{A}.
\end{align}
Since the transmission of interest occurs before the $\mathcal{T}_u(\mathcal{S})$-th time slot, then the $u$-th device is still missing files, i.e., $u \in \tilde{\mathcal{U}}$. Therefore, $u \in \mathcal{I} \Rightarrow \mathcal{I} \neq \varnothing$ which contradicts with the initial assumption of interference-free transmissions.

Finally, given that all devices in the network are equally likely to transmit and that at each time instant a single device from $\mathcal{C}_u$ is allowed to transmit, then all devices in $\mathcal{C}_u$ are also equally likely to transmit.

\subsection{Proof of \lref{l3}} 

Since finding the optimal schedule $S^*$ for the whole recovery phase is intractable, this paper proposes finding the schedule that minimize the probability of increase in the expected completion time at each transmission. Formally, the set of transmitting devices $\mathcal{A}$ and the coded file combination $\kappa_{a}(\mathcal{A})$ are chosen such that:
\begin{align}
\min_{\substack{ \mathcal{A} \in \mathcal{P}(\mathcal{U}) \\ \kappa_a(\mathcal{A}) \in \mathcal{P}(\mathcal{H}_{a})}} \left\{\mathds{P}\left[\max_{u\in \tilde{\mathcal{U}}} \left\{\mathcal{T}_u(t)\right\} > \max_{u\in \tilde{\mathcal{U}}} \left\{\mathcal{T}_u(t-1)\right\}\right] \right\} \nonumber
\end{align}

Clearly, not all devices in $\mathcal{U}$ are able to increase the expected completion time even if they experience a decoding delay for the transmission at time $t$. Let $\mathcal{L}$ be the set of devices that are able to increase the expected completion time at the transmission at time $t$ if they experience a decoding delay. The mathematical definition of this set is given below:
\begin{align}
&\mathcal{L}= \left\{u \in \tilde{\mathcal{U}} \bigg| \mathcal{T}_u(t-1) + \frac{1}{1-\overline{\epsilon}_u\cfrac{|\mathcal{C}_u|-1}{|\mathcal{C}_u|}} \geq  \max_{u^\prime \in \tilde{\mathcal{U}}}\left(\mathcal{T}_{u^\prime}(t)\right)\right\}, \nonumber
\end{align}

Such set $\mathcal{L}$ is called the \emph{critical set} as only devices in this set play a role in the optimization problem and are enable to increase the expected completion time at the transmission at time $t$.

According the definition of $\mathcal{T}_u(t)$ in \eref{eq:3}, devices $u \in \mathcal{L}$ would not increase $\max_{u\in\tilde{\mathcal{U}}}\left\{\mathcal{T}_u(t)\right\}$ only if they do not experience a decoding delay in the transmission at time $t$. Therefore, the probability that the completion time does not increase at time $t$ can be expressed as :
\begin{align}
\mathds{P}&\left[\max_{ u \in \tilde{\mathcal{U}}}\left\{\mathcal{T}_u(t)\right\} = \max_{u \in \tilde{\mathcal{U}} }\left\{\mathcal{T}_u(t-1)\right\}\right] \nonumber \\
& = \mathds{P}\left[\max_{ u \in \mathcal{L}}\left\{\mathcal{T}_u(t)\right\} = \max_{u \in \tilde{\mathcal{U}} }\left\{\mathcal{T}_u(t-1)\right\}\right] \nonumber \\
& =  \mathds{P} \left[d_u(\mathcal{A},\kappa_{a}(\mathcal{A}))=0, \forall u \in \mathcal{L} \right] \nonumber \\
& = \prod_{u \in \mathcal{L}}  \mathds{P} \left[d_u(\mathcal{A},\kappa_{a}(\mathcal{A}))=0 \right]
\end{align}

Hence, the joint online optimization over the set of transmitting devices $\mathcal{A}$ and the coded file combination $\kappa_{a}(\mathcal{A})$ is given by the following expression:
\begin{align}
&\min_{\substack{ \mathcal{A} \in \mathcal{P}(\mathcal{U}) \\ \kappa_a(\mathcal{A}) \in \mathcal{P}(\mathcal{H}_{a})}} \left\{\mathds{P}\left[\max_{u\in \tilde{\mathcal{U}}} \left\{\mathcal{T}_u(t)\right\} > \max_{u\in \tilde{\mathcal{U}}} \left\{\mathcal{T}_u(t-1)\right\}\right] \right\} \nonumber \\
&\min_{\substack{ \mathcal{A} \in \mathcal{P}(\mathcal{U}) \\ \kappa_a(\mathcal{A}) \in \mathcal{P}(\mathcal{H}_{a})}} \left\{ 1 - \mathds{P}\left[\max_{ u \in \tilde{\mathcal{U}}}\left\{\mathcal{T}_u(t)\right\} = \max_{u \in \tilde{\mathcal{U}} }\left\{\mathcal{T}_u(t-1)\right\}\right] \right\} \nonumber \\
&\max_{\substack{ \mathcal{A} \in \mathcal{P}(\mathcal{U}) \\ \kappa_a(\mathcal{A}) \in \mathcal{P}(\mathcal{H}_{a})}} \left\{  \mathds{P}\left[\max_{ u \in \tilde{\mathcal{U}}}\left\{\mathcal{T}_u(t)\right\} = \max_{u \in \tilde{\mathcal{U}} }\left\{\mathcal{T}_u(t-1)\right\}\right] \right\} \nonumber \\
&\max_{\substack{ \mathcal{A} \in \mathcal{P}(\mathcal{U}) \\ \kappa_a(\mathcal{A}) \in \mathcal{P}(\mathcal{H}_{a})}} \left\{ \prod_{u \in \mathcal{L}}  \mathds{P} \left[d_u(\mathcal{A},\kappa_{a}(\mathcal{A}))=0 \right]\right\}
\end{align}

\bibliographystyle{IEEEtran}
\bibliography{references}

\end{document}